\documentclass[journal]{IEEEtran}

\usepackage[latin1]{inputenc}
\usepackage[english]{babel}
\usepackage{amsmath}
\usepackage{amsfonts}
\usepackage{amssymb}
\usepackage{geometry}
\usepackage{graphicx}
\usepackage[x11names,usenames,dvipsnames,svgnames,table,rgb]{xcolor}
\usepackage{hyperref}
\author{Edouard~Pauwels, Amir~Beck, Yonina~C.~Eldar,~\IEEEmembership{Fellow,~IEEE}, Shoham~Sabach
\thanks{The work of Edouard Pauwels was partially supported by the Air Force Office of Scientific Research grant number FA9550-15-1-0500. The research of Amir Beck was partially supported by the Israel Science Foundation Grant 1821/16.}
\thanks{Edouard Pauwels is with the Informatics department (IRIT), Universit\'e Toulouse 3 Paul Sabatier, Toulouse 31062, France (e-mail epauwels@irit.fr).}
\thanks{A. Beck is with the department of Industrial Engineering, Technion--Israel Institute of Technology, Haifa, Israel 32000 (e-mail: becka@ie.technion.ac.il).}
\thanks{Y. C. Eldar is with the department of Electrical Engineering, Technion--Israel Institute of Technology, Haifa, Israel 32000 (e-mail: yonina@ee.technion.ac.il).}
\thanks{S. Sabach is with the department of Industrial Engineering, Technion--Israel Institute of Technology, Haifa, Israel 32000 (e-mail: ssabach@ie.technion.ac.il).}
}
\title{On Fienup Methods for Regularized Phase Retrieval}
\newcommand{\RR}{\mathbb{R}}
\newcommand{\CC}{\mathbb{C}}
\newcommand{\NN}{\mathbb{N}}

\newtheorem{theorem}{Theorem}[section]
\newtheorem{lemma}[theorem]{Lemma}

\newtheorem{definition}{Definition}
\newtheorem{example}{Example}

\newtheorem{corollary}[theorem]{Corollary}
\newenvironment{proof}[1][]{\begin{trivlist}\item{\bf Proof #1:\;}}{\hfill $\Box$\end{trivlist}}
\newcommand{\qed}{\nobreak \ifvmode \relax \else
      \ifdim\lastskip<1.5em \hskip-\lastskip
      \hskip1.5em plus0em minus0.5em \fi \nobreak
      \vrule height0.75em width0.5em depth0.25em\fi}
\date{Draft of \today}
\renewcommand{\Re}{{\rm Re}}
\renewcommand{\Im}{{\rm Im}}
\newcommand{\supp}{{\rm supp}}

\colorlet{FRAME}{yellow!5!white}
\newcommand{\FF}{\mathcal{F}}

\DeclareMathOperator*{\argmin}{argmin}
\newcommand{\prox}{\operatorname{prox}}
\newcommand{\sign}{{\rm sign}}
\newcommand{\bx}{{\bf x}}
\newcommand{\bc}{{\bf c}}

\newcommand{\bw}{{\bf w}}
\newcommand{\by}{{\bf y}}
\newcommand{\bz}{{\bf z}}
\newcommand{\beps}{\boldsymbol{\epsilon}}

\usepackage{soul}
\begin{document}
\maketitle
\begin{abstract}
Alternating minimization, or Fienup methods, have a long history in phase retrieval. We provide new insights related to the empirical and theoretical analysis of these algorithms when used with Fourier measurements and combined with convex priors.  In particular, we show that Fienup methods can be viewed as performing alternating minimization on a regularized nonconvex least-squares problem with respect to amplitude measurements. We then prove that under mild additional structural assumptions on the prior (semi-algebraicity), the sequence of signal estimates has a smooth convergent behaviour towards a critical point of the nonconvex regularized least-squares objective. Finally, we propose an extension to Fienup techniques, based on a projected gradient descent interpretation and acceleration using inertial terms. We demonstrate experimentally that this modification combined with an $\ell_1$ prior constitutes a competitive approach for sparse phase retrieval.
\end{abstract}
\section{Introduction}

Phase retrieval is an old and fundamental problem in a variety of areas within engineering and physics \cite{hurt2001phase,Eldar_review}. Many applications of the phase retrieval problem involve estimation of a signal from the modulus of its Fourier measurements. This problem is ill posed in general, so that uniqueness and recovery typically require prior knowledge on the input, particularly in one-dimensional problems. Here we focus on the estimation of real sparse signals from their Fourier magnitude, a problem which has been treated in several recent works \cite{Kishore:12,shechtman2014gespar,ohlsson2014conditions,SESSY11}.

A longstanding line of algorithms to tackle the phase retrieval problem involve application of the alternating minimization method which alternate between the constraints in time and the Fourier magnitude constraints
\cite{gerchberg1972practical,fienup_reconstruction_1978,fienup1982phase}. These methods were pioneered by the work of Gerchberg
and Saxton and later extended by Fienup; see \cite{bauschke2002phase} for an optimization point of view on these techniques and a rich historical perspective.
Alternating minimization approaches have also been recently applied to phase retrieval from random measurements \cite{netrapalli2015phase}. The main advantage of this class of algorithms is their simplicity and scalability.

A more recent approach to phase retrieval is to formulate the recovery as a smooth nonconvex least-squares estimation problem and use dedicated techniques to estimate the signal using continuous optimization algorithms that guarantee convergence to stationary points. The GESPAR algorithm \cite{shechtman2014gespar} is an example of this approach which is based on the Gauss-Newton method coupled with sparsity priors.
For phase retrieval with random measurements, gradient descent methods have been proposed and analyzed such as Wirtinger flow \cite{candes2014phase} and truncated amplitude flow \cite{WGE16}. Both treat least-squares objectives where Wirtinger flow measures the loss with respect to the squared-magnitude of the measurements while the amplitude flow approach performs a truncated gradient descent on an amplitude objective.
Another line of work suggests the use of matrix lifting and semidefinite programming based relaxations \cite{waldspurger2015phase,ohlsson2012compressive,jaganathan2012recovery,CESV13,SESSY11}. These techniques are limited by the size of problems that can be tackled using available numerical solvers.

Our main contribution is to propose a new look at alternating minimization algorithms for phase retrieval in the context of Fourier measurements and convex priors. We refer collectively to these techniques as Fienup methods.
The use of Fourier measurements is less flexible than general measurements and is less suited for statistical analysis. On the other hand, the Fourier transform has very strong structure which allows for richer algorithmic constructions and analysis.

 As a first step we provide two new interpretations of Fienup algorithms. First we show that these techniques are naturally linked to a nonsmooth nonconvex least-squares problem with respect to an amplitude objective.
 Fienup approaches can then be understood as majorization-minimization methods for solving this problem. Second, we demonstrate that Fienup algorithms can be viewed as a projected gradient descent scheme to minimize a smooth convex objective function over a nonconvex constraint set. This observation  allows to characterize the behaviour of the algorithm and develop extensions based on known ideas for accelerating gradient methods using inertial terms \cite{beck2009fast}. We then specialize these results to the case of $\ell_0$ and $\ell_1$ priors, leading to a new inertial gradient scheme, which we refer to as FISTAPH: FISTA for PHase retrieval.

On the theoretical side, we show that if the convex prior is well structured (semi-algebraic or more generally representable), then the sequence of signal estimates produced by Fienup has a smooth convergence behaviour. Recall that, broadly speaking, an object is said to be semi-algebraic if it can be represented by systems of polynomial inequalities. The notion of smooth convergence is a very desirable property, even more in nonconvex settings where it is usually not possible to obtain global convergence estimates. The convergence analysis follows well established techniques from tame optimization \cite{attouch2013convergence, bolte2014proximal}. These techniques build upon the Kurdyka-\L ojasiewicz (KL) property which holds for many classes of functions \cite{loja1963propriete,kurdyka1998gradients,bolte2007lojasiewicz,bolte2007clarke}. We then provide numerical experiments based on synthetic problems to compare Fienup with $\ell_0$ and $\ell_1$ priors, GESPAR \cite{shechtman2014gespar}, Wirtinger flow (or gradient) methods \cite{candes2014phase} with $\ell_0$ and $\ell_1$ priors and FISTAPH. Numerical results suggest that the latter combined with an $\ell_1$ prior constitutes a very competitive alternative for sparse phase retrieval.

The rest of the paper is organized as follows. Section \ref{sec:problem_form} introduces our notation and states the problem of interest more formally. We also introduce several mathematical definitions that are required for the rest of the paper and review the numerical algorithms that are used in subsequent sections. Section \ref{sec:interpretations} describes our  characterization of Fienup methods in the context of phase retrieval from Fourier measurements with convex priors. We detail the relation of Fienup with a nonsmooth nonconvex least-squares problem as well as its interpretation as projected gradient descent. Our main convergence result and our new FISTAPH algorithm are presented in Section \ref{sec:convergence}. Simulation results are provided in Section \ref{sec:experiments}.

\section{Problem Formulation and Mathematical Background} \label{sec:problem_form}
\subsection{Notation}
Throughout the paper vectors are denoted by boldface letters. For a vector $\bx \in \CC^n$, $\bx[i]$ is the $i$-th entry of $\bx$, $i=1,2,\ldots, n$ and $\supp(\bx)$ is the support of $\bx$, namely, the set $\left\{ i = 1, 2, \ldots n;\; \bx[i] \neq 0 \right\}$. Furthermore, $\|\bx\|_0$ counts the number of nonzero entries of the vector $\bx$: $\|\bx\|_{0} = |\supp(\bx)|$ and $\|\bx\|_p$ denotes the $\ell_p$ norm of $\bx$ for $p \in \RR_+$. The notations $|\cdot|$, $\Re(\cdot)$, $\Im(\cdot)$ and $\bar{\cdot}$ describe the modulus, real part, imaginary part and complex conjugate, respectively, defined over the field of complex numbers. If their argument is a vector, then they should be understood component-wise. Similarly, basic operations, e.g. powers, are taken component-wise when applied to vectors. For $\bx \in \CC^n$ and $N \in \NN$, $\FF(\bx, N) \in \CC^N$ is the vector composed of the $N$ first coefficients of the discrete Fourier transform of $\bx$ (obtained by  zero padding if $n <N$). For simplicity, we use the shorthand notation $\FF (\bx) = \FF(\bx,n)$ to denote the standard discrete Fourier transform of $\bx \in \CC^n$ and $\FF^{-1}$ to denote its inverse. For a set $S$, $\delta_S \colon S \to \RR \cup \{+\infty\}$ is the indicator function of $S$ ($0$ if its argument is in $S$, $+\infty$ otherwise) and $P_S$ denotes the Euclidean orthogonal projection onto the set $S$.
\subsection{Phase Retrieval}
Given $\bx_0 \in \RR^n$, we consider the data acquisition process
\begin{align}
	\label{eq:dataAcq}
	\bc = |\FF(\bx_0)| + \beps,
\end{align}
where $\beps \in \RR^n$ is an unknown vector of errors. In the rest of the paper, we actually assume that $\bc$ has positive entries (it is always possible to set the potential negative entries of $\bc$ to zero). The phase retrieval problem consists of producing an estimate $\hat{\bx} \in \RR^n$ of $\bx_0$ based solely on the knowledge of $\bc$ given by (\ref{eq:dataAcq}).

As mentioned in the introduction, phase retrieval of one-dimensional vectors from Fourier measurements requires the use of prior knowledge. We focus on support and sparsity inducing priors. For $J \subseteq \left\{ 1,2,\ldots,n \right\}$, we define the set $X_J = \left\{ \bx \in \RR^n;\; \supp(\bx) \subseteq J \right\}$. The prior function that we use will be denoted by $g \colon \RR^n \to \RR \cup \{+\infty\}$. We focus on the following priors (for a given $J$):
\begin{itemize}
	\item $g \colon \bx \mapsto \|\bx\|_0 + \delta_{X_J}(\bx)$, or $\ell_0$-based nonconvex prior.
	\item $g \colon \bx \mapsto \|\bx\|_1 + \delta_{X_J}(\bx)$, or $\ell_1$-based convex prior.
\end{itemize}
In the experimental section, we compare between these two classes of priors. The algorithmic derivations in this paper will be made under the assumption that $g$ is proper and lower semicontinuous, and the main convergence result (c.f. Theorem \ref{th:convergence}) will require in addition convexity of $g$.
                In order to efficiently implement the proposed algorithm, we need to focus on priors for which the proximity operator \cite{moreau1965proximite} is easy to compute. We provide several examples of such priors in Section \ref{sec:tools}.

In the rest of the paper, $\bc \in \RR^n_+$ denotes modulus measurements which are assumed to be given, fixed and obtained through (\ref{eq:dataAcq}). Given $\bc \in \RR^n_+$, we define $Z_{\bc} = \left\{ \bz \in \CC^n;\; |\FF(\bz)| = \bc \right\}$ as the set of values $\bz$ that could have produced $\bc$ (ignoring the noise). To estimate $\bx_0$, we consider the regularized least-squares problem
\begin{align}
	\label{eq:mainProblem}
	\min_{\bx\in \RR^n, \bz\in Z_{\bc}} \frac{1}{2}\|\bx - \bz\|_2^2 + g(\bx),
\end{align}
where $g$ encodes our prior knowledge. Our algorithmic approach consists of employing an alternating minimization method, or one of its variants, to solve the above formulation.

\subsection{Prior Algorithms for Phase Retrieval}
\label{sec:algorithms}
We briefly review several existing algorithms for phase retrieval that will be used in our experiments in Section \ref{sec:experiments}.

One approach to sparse phase retrieval is the GESPAR algorithm which is based on the damped Gauss-Newton method in conjunction with an $\ell_0$ prior \cite{shechtman2014gespar}. Damped Gauss-Newton allows to solve smooth, nonlinear least-squares problems. The work of \cite{candes2014phase} is based on the notion of Wirtinger derivatives to treat the same smooth least-squares problem as GESPAR. The notion of Wirtinger derivative is needed since the objective is not differentiable (holomorphic) as a function of complex variables (see \cite{candes2014phase} for details). In the case of real valued functions of real variables, the Wirtinger derivative reduces to a standard gradient (up to a constant multiplicative factor).
An obvious extension of these types of methods is the use of proximal decomposition, or forward-backward methods which consist in alternating a gradient step on the smooth part of the objective with a proximal step on the nonsmooth part \cite{beck2009fast,beck2009gradient}. This is the approach that we use in the numerical experiments to treat phase retrieval with priors.

Finally, we consider alternating minimization methods that are the main focus of this work. This approach consists of solving (\ref{eq:mainProblem}) by applying the alternating minimization algorithm. The special structure of the problem allows to perform each partial minimization efficiently. In particular, the projection onto $Z_{\bc}$ is easy, as described below in (\ref{eq:solutionZ}). These types of methods are also referred to as Fienup algorithms. A deeper interpretation of this approach is given in Section \ref{sec:interpretations}.

\subsection{Tools from Convex and Nonsmooth Analysis}
\label{sec:tools}
Throughout the paper, our results will be based on tools from convex and nonsmooth analysis which we review here.

The gradient of a differentiable function $f$ is denoted by $\nabla f$. This concept admits extensions to nonsmooth analysis; the subgradient of a nonsmooth function $g$ is denoted by $\partial g$. For convex functions, subgradients correspond to tangent affine lower bounds. This definition no longer holds for nonconvex functions. In this case, the proper understanding of subgradients involves much more machinery which will not be discussed here. We only consider the notion of a Fr\'echet critical point which generalizes classical first order criticality for differentiable functions (see \cite{rockafellar1998variational}).
\begin{definition}[Fr\'{e}chet critical point]
	\label{def:regularCriticalValue}
	Let $S \subseteq \RR^n$ be a closed set and $f\colon \RR^n \to \RR$ be a lower semicontinuous function. We say that $\bar{\bx} \in S$ is a Fr\'{e}chet critical point of the problem
	\begin{align*}
		\min_{\bx \in S} f(\bx)
	\end{align*}
	if
\begin{align*}
		\underset{
			\begin{subarray}{c}
				x \to \bar{x} \\
				x \neq \bar{x} \\
				x \in S
			\end{subarray}
		}{\lim\inf}& \quad \frac{f(x) - f(\bar{x})}{\|x - \bar{x}\|} \geq 0.
	\end{align*}
%for any $\tilde{\bx} \in S$, it holds that
%	\begin{align*}
%		f(\tilde{\bx}) \geq f(\bar{\bx}) + o\left( \|\tilde{\bx} - \bar{\bx}\|_2 \right).
%	\end{align*}
	In other words, the negative variations of $f$ in $S$ around $\bar{x}$ are negligible at the first order.
\end{definition}
We will also heavily use the notion of the proximity operator of a function.
\begin{definition}[Proximity operator]
	For a nonsmooth function $g: \RR^n \to \RR \cup \{+\infty\}$, the (potentially multivalued) proximity operator is denoted by $\prox_g$ and defined by
	\begin{align}
	\label{eq:defProx}
	\prox_g(\bx) \equiv  \argmin_{\by \in \RR^n} \left \{ \frac{1}{2}\|\bx - \by\|_2^2 + g(\by)\right\}.
\end{align}
\end{definition}
Note that when $g$ is proper lower semicontinuous and convex, this operator is single valued.

We next provide a few examples of such functions with their proximity operators; many more can be found, for example, in \cite{combettes2011proximal}.
\begin{example}[Proximity operators]\hfill

	\begin{itemize}
		\item {\bf Support prior:} If $C \subseteq \RR^n$ is a closed convex set, then $\prox_{\delta_C}$ is the Euclidean projection onto $C$. This can be used for example to encode knowledge about the support of the signal $\bx_0$ by choosing $C = X_J$ for some $J \subseteq \left\{ 1, 2, \ldots, n \right\}$. In this case, the projection simply consists in setting the coefficients $\bx[i]$ to $0$ for $i \not\in J$.
		\item {\bf Sparsity prior:} If $g$ is the $\ell_1$ norm, then the proximal operator is the soft thresholding operator.\footnote{\label{footnote_soft} The soft thresholding operator is given by ${\mathcal T}_{\alpha}(\bx)_i = {\rm sgn}(x_i)\max\{ |x_i|-\alpha,0\}$. If $g(\bx) = \lambda \|\bx\|_1$ for some $\lambda>0$, then $\prox_{g}(\bx) = {\mathcal T}_{\lambda}(\bx)$.}  This can be combined with support information prior by first setting the coefficients outside of the support to $0$ and then applying the soft thresholding operator.
		\item {\bf Change of basis:} Suppose that $D$ is an $n \times n'$ real matrix such that its columns form an orthonormal family, that is $D^T D$ is the identity in $\RR^{n'}$. Suppose that $\tilde{g}\colon \RR^{n'} \to \RR$ is a lower semicontinuous convex function and let $g(\bx)=\tilde{g}(D^T\bx)$. In this case, we have $\prox_{g}( \bx)= \bx + D\left( \prox_{\tilde{g}}(D^T\bx) - D^T\bx \right)$ (see \cite[Table 1]{combettes2011proximal}). This allows to express priors in different orthonormal bases, such as wavelets for example.
	\end{itemize}
\end{example}
It is also worth mentioning that the proximity operator is efficiently computable for some nonconvex priors. For example, if $g = \delta_C$ where $C = \left\{ \bx \in \RR^n;\; \|\bx\|_0 \leq k \right\}$, then the proximity operator is obtained by setting the $n-k$ lowest coefficients (in absolute value) to $0$. This can also be combined with support information.

\section{Fienup, Majorization-Minimization and Projected Gradient}
\label{sec:interpretations}

In this section we expand on the alternating minimization approach to (\ref{eq:mainProblem}) leading to the Fienup family of algorithms. For this section, the prior term $g$ in (\ref{eq:mainProblem}) is taken to be a general proper lower semicontinuous   function. We begin by describing the algorithm and then provide two interpretations of it.
\subsection{Alternating Minimization Algorithm}
The alternating minimization algorithm applied to problem (\ref{eq:mainProblem}) is explicitly written below.

\bigskip
\fcolorbox{black}{Ivory2}{\parbox{7cm}{{\bf Alternating Minimization (Fienup)} \\
	{\bf Initialization.} $\bx^{0} \in \RR^n$\\
  {\bf General Step.} For $k \in \NN$,
	\begin{align}
		\label{eq:altMinPres}
		\bz^{k+1} &\in \argmin_{\bz \in Z_\bc} \frac{1}{2}\|\bx^k - \bz\|_2^2,\nonumber\\
		\bx^{k+1} &\in     \argmin_{\bx \in \RR^n} \frac{1}{2}\|\bx - \bz^{k+1}\|_2^2 + g(\bx).
	\end{align}
	}
}
\bigskip

The main interest in this scheme is that both partial minimization steps in (\ref{eq:altMinPres}) can be carried out efficiently whenever $g$ is ``proximable", meaning that its prox  (or a member in its prox) is easily computed. First consider, in (\ref{eq:altMinPres}), the partial minimization in $\bz$ with $\bx \in \RR^n$ being arbitrary but fixed. This minimization amounts to computing $P_{Z_{\bc}}(\bx^k)$, the orthogonal  projection of $\bx^k$ onto $Z_\bc$. For a given $\bx \in \RR^n$, all the members in $P_{Z_{\bc}}(\bx)$ are of the form $\bz = \FF^{-1}(\hat{\bz})$, where for $j=1,2,\ldots, n$, we have ($i=\sqrt{-1}$ in the equation below)
\begin{align}
	\label{eq:solutionZ}
	\hat{\bz}[j] =
	\begin{cases}
		\bc[j]\frac{\FF(\bx)[j]}{|\FF(\bx)[j]|},		
		&\text{if } |\FF(\bx)[j]| \neq 0, \\
		\bc[j]\mathrm{e}^{i\theta_j}, & \text{for an arbitrary } \theta_j\, \text{ otherwise}.
	\end{cases}
\end{align}
Next, we treat the subproblem in (\ref{eq:altMinPres}) of minimizing with respect to $\bx$ where $\bz \in \CC^n$ is arbitrary but fixed. The partial minimization in $\bx$ is given by the expression
\begin{align}
	\label{eq:solutionX}
	\text{argmin}_{\bx\in \RR^n} \left \{ \frac{1}{2}\|\bx - \bz\|_2^2 + g(\bx)\right \} = \prox_{g}({\rm Re}(\bz)),
\end{align}
where ${\rm Re}$ is the real part taken component-wise. We have used the definition of the proximity operator of $g$ given in (\ref{eq:defProx}). When this operator is easy to compute, each step of the algorithm can be carried out efficiently.

The iterations of the alternating minimization method are summarized as follows:
\begin{align}
	\label{eq:itSummary}
	\bz^{k+1} &\in P_{Z_\bc}(\prox_{g}({\rm Re}(\bz^k))), \nonumber\\
	\bx^{k+1} &\in \prox_{g}({\rm Re}(P_{Z_\bc}(\bx^k))).
\end{align}
We now consider several special cases of (\ref{eq:itSummary}):
\begin{itemize} \item If $g=0$, then $\prox_g$ is the identity and we recover the original algorithm from Fienup \cite{fienup1982phase}, or alternating projection \cite{bauschke2002phase}.

\bigskip
 \fcolorbox{black}{Ivory2}{\parbox{6cm}{{\bf Fienup} \\
	{\bf Initialization.} $\bx^{0} \in \RR^n$.\\
  {\bf General Step.} For $k \in \NN$,
$$  \bx^{k+1} ={\rm Re}(P_{Z_\bc}(\bx^k)).$$
}}
\bigskip

The convergence result given in Theorem \ref{th:convergence} also holds in this case since constant functions are convex and continuous.
\item If $g(\bx) = \lambda \|\bx\|_1$ for some $\lambda>0$, then $\prox_{g} = {\mathcal T}_{\lambda}$, where ${\mathcal T}_{\lambda}$ is the soft thresholding operator (see footnote on page \pageref{footnote_soft}). We refer to the resulting algorithm as ``AM L1".\\

\bigskip
    \fcolorbox{black}{Ivory2}{\parbox{6cm}{{\bf AM L1} \\
	{\bf Initialization.} $\bx^{0} \in \RR^n$, $\lambda>0$\\
  {\bf General Step.} For $k \in \NN$,
$$  \bx^{k+1} = {\mathcal T}_{\lambda}({\rm Re}(P_{Z_\bc}(\bx^k))).$$
}}
\bigskip
\item If $g=\delta_{C_K}$, where $C_K$ is the set of all $K$-sparse vectors, $C_K = \{ \bx  \in \RR^n: \|\bx\|_0 \leq K\}$, then $\prox_g = P_{C_K}$ is the so-called hard thresholding operator. This operator outputs a vector which is all zeros except for the largest $K$ components (in absolute values) of its input vector which are kept the same. The hard thresholding operator is multivalued and the resulting algorithm, which we term ``AM L0" picks an arbitrary point in its range.\\

\bigskip
     \fcolorbox{black}{Ivory2}{\parbox{6cm}{{\bf AM L0} \\
	{\bf Initialization.} $\bx^{0} \in \RR^n$, $K \in \NN$\\
  {\bf General Step.} For $k \in \NN$,
$$  \bx^{k+1} \in P_{C_K}({\rm Re}(P_{Z_\bc}(\bx^k))).$$
}}
\bigskip
    \end{itemize}
\subsection{Majorization-Minimization Interpretation}
In this section, we focus on partial minimization in $\bz$. We show that the value of this partial minimization leads to a least-squares objective. This allows us to interpret the Fienup algorithm as a majorization-minimization process on this least-squares function. For the rest of this section, for any $\bx \in \RR^n$, we denote by $\bz(\bx)$ an arbitrary but fixed member of $P_{Z_{\bc}}(\bx)$.
\subsubsection{Partial Minimization in $\bz$}
The following lemma provides a connection between partial minimization in $\bz$ and the evaluation of a nonsmooth least-squares objective.
\begin{lemma}
	\label{lem:minZModel2}
	For any $\bx \in \RR^n$, we have
	\begin{align}
		\label{eq:mainIdentity}
		\min_{\bz \in Z_{\bc}} \frac{1}{2}\|\bx - \bz\|_2^2 = \frac{1}{2n}\||\FF(\bx)| - \bc\|_2^2.
	\end{align}
\end{lemma}
\begin{proof}
An optimal solution of the minimization problem is given by $\bz = \FF^{-1}(\hat{\bz})$ where $\hat{\bz}$ has the form (\ref{eq:solutionZ}). Now,
\begin{align*}
	\min_{\bz \in Z_\bc} \frac{1}{2}\|\bx - \bz\|_2^2 &= \frac{1}{2}\|\bx - \FF^{-1}(\hat{\bz})\|_2^2\\
	&=\frac{1}{2}\|\FF^{-1}(\FF(\bx) -\hat{\bz}))\|_2^2\\
	&=\frac{1}{2n}\|\FF(\bx) - \hat{\bz}\|_2^2.
\end{align*}
Using the expression of $\hat{\bz}$ in (\ref{eq:solutionZ}), we have for all $j=1, 2, \ldots, n$,
\begin{align*}
	|\FF (\bx)[j] - \hat{\bz}[j]| =
	\begin{cases}
		||\FF (\bx)[j]| - \bc[j]|,		
		&\text{if } |\FF(\bx)[j]| \neq 0, \\
		\bc[j], & {\rm otherwise}.
	\end{cases}
\end{align*}
Putting everything together,
\begin{align}
	\label{eq:mainIdentity}
	\min_{\bz \in Z_\bc} \frac{1}{2}\|\bx - \bz\|_2^2  =\frac{1}{2n}\|\FF(\bx) - \hat{\bz}\|_2^2   = \frac{1}{2n}\||\FF(\bx)| - \bc\|_2^2,
\end{align}
which completes the proof.
\end{proof}
A direct consequence of Lemma \ref{lem:minZModel2} is the following corollary that connects between problem (\ref{eq:mainProblem}) and a regularized nonlinear least-squares problem.
\begin{corollary} The pair $(\bx,\bz)$ is an optimal solution of problem (\ref{eq:mainProblem}) if and only if $\bx$ is an optimal solution of
\begin{equation} \label{defH} \min \left \{ F(\bx) \equiv \frac{1}{2n}\||\FF(\bx)| - \bc\|_2^2+g(\bx) \right \}, \end{equation}
and $\bz = \FF^{-1}(\hat{\bz})$, where $\hat{\bz}$ is of the form given in (\ref{eq:solutionZ}).
\end{corollary}
Note that in (\ref{defH}), the least-squares objective is defined with respect to the amplitude $|\FF(\bx)|$ and not the magnitude-squared $|\FF(\bx)|^2$.
For random measurements, it has been shown in \cite{WGE16} that the amplitude objective leads to superior performance over the standard magnitude-squared approach.
\subsubsection{Fienup as Majorization-Minimization}
\label{sec:majMin}
In order to understand further the connection with the Fienup algorithm, we define the following auxiliary function:
\begin{align*}
	h (\bx,\by) \equiv \frac{1}{2}\|\by - \bz(\bx)\|_2^2 + g(\by).
\end{align*}
Now, for any $\bx \in \RR^n$, using Lemma \ref{lem:minZModel2}, we have the following properties (recalling the definition of $F$ in (\ref{defH})):
\begin{align}
	\label{eq:majorizationMinimization}
	h(\bx, \by) &= \frac{1}{2}\|\by - \bz(\bx)\|_2^2 + g(\by) \\
	&\geq \frac{1}{2}\|\by - \bz(\by)\|_2^2 + g(\by) = F(\by),\quad \forall \by \in \RR^n,\nonumber\\
	h(\bx, \bx) &= \frac{1}{2}\|\bx - \bz(\bx)\|_2^2 + g(\bx) = F(\bx).\nonumber
\end{align}
In other words, using the convexity of $g$, we have that $h(\bx, \cdot)$ is a $1$-strongly convex global upper bound on the objective $F$. Computing this upper bound amounts to performing partial minimization over $\bz$ in (\ref{eq:mainProblem}). Minimizing the upper bound $h(\bx, \by)$ in $\by$ corresponds to partial minimization over $\bx$ in (\ref{eq:mainProblem}). The upper bound is tight in the sense that we recover the value of the objective at the current point, $h(\bx,\bx) = F(\bx)$. Therefore the alternating minimization algorithm is actually a majorization-minimization method for the nonsmooth least-squares problem
\begin{align}
	\label{eq:mainProblem2Bis}
	\min_{\bx \in \RR^n} \frac{1}{2n}\||\FF(\bx)| - \bc\|_2^2 + g(\bx).
\end{align}
The steps presented in (\ref{eq:altMinPres}) can then be summarized as follows:
\begin{align*}
	\bx^{k+1} &= \argmin_\by \;h(\bx^k, \by) \\
	&= \prox_{g}({\rm Re}(\bz(\bx^k))) \\
	&=\prox_{g}({\rm Re}(P_{Z_\bc}(\bx^k))),
\end{align*}
which is exactly the mapping given in (\ref{eq:itSummary}).

\subsection{Projected Gradient Descent Interpretation}
\label{sec:projGrad}
We now provide an additional interpretation of the alternating minimization algorithm as a projected gradient method for an optimization problem related to (\ref{eq:mainProblem}) which consists of a smooth convex objective and a nonconvex constraint set. This interpretation is valid whenever $g$ is assumed to be proper lower semicontinuous and convex.

For any $\bx \in \RR^n$ and $\bz \in \CC^n$, we can write (\ref{eq:mainProblem}) as $\|\bx - \bz\|_2^2 = \|\bx - \Re(\bz)\|_2^2 + \|\Im(\bz)\|_2^2$. To move from complex numbers to real numbers, we set $\bw_1 = \Re(\bz)$ and $\bw_2=\Im(\bz)$. Defining a new constraint set $\tilde{Z}_{\bc} = \left\{(\bw_1, \bw_2) \in \RR^n \times \RR^n:\; \bw_1 + i \bw_2 \in Z_{\bc}  \right\}$, problem (\ref{eq:mainProblem}) can be  equivalently rewritten in the form
\begin{align}
	\label{eq:mainProblem3}
	\min_{\bx \in \RR^n, (\bw_1, \bw_2)\in \tilde{Z}_{\bc}} \left \{ \frac{1}{2}\|\bx - \bw_1\|_2^2 + \frac{1}{2} \|\bw_2\|_2^2 + g(\bx)\right \}.
\end{align}

Minimizing first w.r.t. $\bx$, (\ref{eq:mainProblem3}) reduces to the following minimization problem in $\bw_1,\bw_2$:
\begin{align}
	\label{eq:mainProblem3:a}
	\min_{(\bw_1, \bw_2) \in \tilde{Z}_{\bc}} 	\left \{ H(\bw_1,\bw_2) \equiv G(\bw_1)+\frac{1}{2}\|\bw_2\|^2 \right \},
\end{align}
where
$$ G(\bw_1) \equiv \min_{\bx \in\RR^n} \left \{\frac{1}{2}\|\bw_1 - \bx\|_2^2 + g(\bx)\right \}.$$
%
%Minimizing
%We define the following functions:
%\begin{align}
%	\label{eq:obj}
%	G &\colon \bw_1 \mapsto \min_{\bx \in\RR^n} \left \{\frac{1}{2}\|\bw_1 - \bx\|_2^2 + g(\bx)\right \},\nonumber\\
%	H & \colon (\bw_1,\bw_2) \mapsto  G(\bw_1) + \frac{1}{2} \|\bw_2\|_2^2.
%\end{align}
The following result allows us to relate the gradient of $H$ to the optimization primitives used in the alternating minimization method.
\begin{lemma}
	\label{lem:gradH}
Assume that $g$ is proper, lower semicontinuous and convex. Then the function $H$ is continuously differentiable, its gradient is $1$-Lipschitz and can be expressed as	
\begin{equation}  \label{345}		\nabla 	H(\bw_1,\bw_2) = \left ( \bw_1 - \prox_g(\bw_1),  \bw_2\right ).\end{equation}
%		=
%		\left(
%		\begin{array}{c}
%			\prox_{g^*}(\bw_1)\\
%			\bw_2
%		\end{array}
%		\right).
%	\end{align*}
\end{lemma}
\begin{proof}
	From Moreau \cite[Proposition 7.d]{moreau1965proximite}, we know that $G$ is differentiable and $\nabla G(\bx) = \bx - \prox_g (\bx) = \prox_{g^*}(\bx)$, where $g^*$ is the conjugate function of $g$, which is convex. The computation of the gradient of $H$ is then immediate. We can use the fact that proximity operators of convex functions are nonexpansive \cite[Proposition 5.b]{moreau1965proximite} to verify that $\nabla H$ is $1$-Lipschitz. Indeed, for any $(\bw_1, \bw_2)$ and $(\tilde{\bw}_1, \tilde{\bw}_2)$, we have
	\begin{align*}
		&\;\|\nabla 	H(\bw_1,\bw_2) - \nabla 	H(\tilde{\bw}_1,\tilde{\bw}_2)\|_2^2 \\
		=&\;\|\prox_{g^*}(\bw_1) - \prox_{g^*}(\tilde{\bw}_1)\|_2^2 + \|\bw_2 - \tilde{\bw}_2\|_2^2 \\
		\leq&\; \|\bw_1 - \tilde{\bw}_1\|_2^2 + \|\bw_2 - \tilde{\bw}_2\|_2^2 \\
		=&\; \|(\bw_1,\bw_2) - (\tilde{\bw}_1,\tilde{\bw}_2)\|_2^2,
	\end{align*}
completing the proof.
\end{proof}

%Problem (\ref{eq:mainProblem}) is equivalently stated in the following form.
%which is the minimization of a convex function over a nonconvex set.

Consider applying projected gradient descent to solve (\ref{eq:mainProblem3:a}). From Lemma \ref{lem:gradH}, we can use a step size of magnitude $1$. In this case, taking into account the form of the gradient given in (\ref{345}), we obtain that the general update step takes the form
%\begin{align*}
%	\bw_1^{k+1} &= \bw_1^k - \nabla_{\bw_1} H(\bw_1^k,\bw_2^k) = \prox_g(\bw_1^k)\\
%	\bw_2^{k+1}&= \bw_2^k - \nabla_{\bw_2} H(\bw_1^k,\bw_2^k) = 0.
%\end{align*}
%We can add the projection on $\tilde{Z}_{\bc}$ to get the projected gradient mapping for problem (\ref{eq:mainProblem3}).
\begin{eqnarray*}
	&&(\bw_1^{k+1},\bw_2^{k+1})\\
	&=& P_{\tilde{Z}_{\bc}}((\bw_1^k,\bw_2^k) - \nabla H(\bw_1^k,\bw_2^k))\\
&=&P_{\tilde{Z}_{\bc}}((\bw_1^k,\bw_2^k) - (\bw_1^k-\prox_g(\bw_1^k),\bw_2^k))\\
&=&  P_{\tilde{Z}_{\bc}}(\prox_g(\bw_1^k),0).
\end{eqnarray*}
We now go back to the complex domain by setting $\bz = \bw_1 + i \bw_2$. Note that projecting $(\bw_1, \bw_2)$ onto $\tilde{Z}_{\bc}$ is equivalent to projecting $\bz$ onto $Z_{\bc}$. With this notation, the iterations of projected gradient descent can be summarized by the following iteration mapping (on complex numbers):
\begin{align*}
	\bz^{k+1} &= P_{Z_{\bc}}(\prox_g(\Re(\bz^k))),
\end{align*}
which is exactly the same as (\ref{eq:itSummary}). Therefore, the Fienup algorithm is equivalent to projected gradient descent with unit stepsize applied to the formulation (\ref{eq:mainProblem3:a}). Note that from the point of view of nonsmooth analysis, problem (\ref{eq:mainProblem3:a}) is much better behaved than (\ref{eq:mainProblem2Bis}).
\section{Consequences and Extensions}
\label{sec:convergence}
The interpretations of Section \ref{sec:interpretations} can be used to analyze the convergence of alternating minimization applied to problem (\ref{eq:mainProblem}) and to offer extensions of the method.
\subsection{Convergence Analysis}
Our main convergence result is given in the following theorem. Recall that a function is semi-algebraic if its graph can be defined by combining systems of polynomial equalities and inequalities (for example, the $\ell_1$ norm is semi-algebraic).
\begin{theorem}
	\label{th:convergence}
	Assume that $g$ is proper, lower semicontinuous, convex and semi-algebraic. Then the sequence $\{\bx^k,\bz^k\}_{k \in \NN}$ generated by the alternating minimization algorithm satisfies the following:
	\begin{itemize}
		\item[(i)] It holds that $\sum_{k\geq 0} \|\bx^{k+1} - \bx^k\|_2 < + \infty$ and the sequence $\{\bx^k\}_{k \in \NN}$ converges to a point $\bx^* \in \RR^n$.
		\item[(ii)] For any accumulation point $\bz^*$ of $\{\bz^k\}_{k \in \NN}$, $(\bx^*, \bz^*)$ is a Fr\'echet critical point of problem (\ref{eq:mainProblem}) and $(\bw_1^*,\bw_2^*)=(\Re(\bz^*),\Im(\bz^*))$ is a Fr\'echet critical point of problem (\ref{eq:mainProblem3:a}).
	\end{itemize}
\end{theorem}
The proof is quite technical and is given in Appendix \ref{sec:proof}. The semi-algebraic assumption on $g$ can be relaxed to representability in o-minimal structures over the real field, see \cite{Dries-Miller96}. Therefore, the proposed result actually applies to much more general regularizers. For example, using boundedness of the feasible set in (\ref{eq:mainProblem3:a}), the same result holds if $g$ is analytic (see the dicussion in \cite[Section 5]{bolte2014majorization}). The arguments build upon a nonsmooth variant of the celebrated Kurdyka-\L ojasiewicz (KL) property \cite{loja1963propriete,kurdyka1998gradients,bolte2007lojasiewicz,bolte2007clarke}. Note that direct application of the results of \cite{attouch2013convergence,bolte2014proximal} to projected gradient descent or the results of \cite{bolte2014majorization,pauwels2016value} to majorization minimization is not possible here.

The most important implication of Theorem \ref{th:convergence} is that the sequence of estimated signals converges smoothly to a point which satisfies certain optimality conditions related to problems (\ref{eq:mainProblem}), and (\ref{eq:mainProblem3:a}). This is a departure from standard convergence results that are only able to guarantee that accumulation points of the generated sequence of iterates satisfy certain optimality conditions. It is important to underline that the result is global: it holds for any initialization of the algorithm and does not require any regularity assumption beyond semi-algebraicity and convexity of $g$. This is in contrast with local convergence results which are typical for alternating projection methods \cite{lewis2009local,bauschke2013restricted} that are applicable when the prior term $g$ is an indicator function.
\subsection{Acceleration and Momentum Term}
A benefit of the interpretation of alternating minimization as a projected gradient method is that it allows to propose new variants inspired by known extensions for projected gradient algorithms. In this section we focus on the incorporation of an inertial term that results in an alternating minimization scheme that includes a momentum term. This line of research has a long history in optimization, starting with the development of the heavy-ball method \cite{polyak1964some} which inspired an optimal first order scheme for convex optimization developed by Nesterov \cite{netrapalli2013phase}, and its extension to convex composite problems with the FISTA method \cite{beck2009fast}. Although this last technique was proposed and analyzed only in the context of convex optimization, we consider its application in our nonconvex constrained problem since it empirically provides interesting results. The resulting algorithm is referred to as FISTAPH, and is described as follows.

\bigskip
\fcolorbox{black}{Ivory2}{\parbox{7cm}{{\bf FISTAPH: FISTA for Phase retrieval} \\
	{\bf Initialization.} $\bz^{0} \in Z_{\bc}$ and $\alpha^{k} \in \left[0 , 1\right)$ for all
		$k \in \NN$. Set $\by^0 = \bz^0$ and $\bz^{-1}=\bz^0$.\\
  {\bf General Step.} For $k \in \NN$,
	\begin{itemize}
		\item $\bz^{k+1} \in P_{Z_{\bc}}(\prox_g(\Re(\by^k)))$.
		\item $\by^{k+1} = \bz^k + \alpha^k\left( \bz^k - \bz^{k-1} \right)$.
	\end{itemize}
	}
}
\bigskip

If $\bz^m$ is the last produced iteration, then the output of the algorithm is $\hat{\bx}=\prox_g(\Re(\bz^m))$. A typical choice for the weight sequence is $\alpha^k = \frac{k-1}{k+2}$.
 The question of the convergence of the iterates produced by this method in nonconvex settings is an interesting topic to explore in future research. We may also further consider monotone variants of similar types of methods, see e.g. \cite{beck2009gradient}.

In the numerical experiments we employ FISTAPH in the setting where $g(\bx) = \lambda \|\bx\|_1$ for some $\lambda>0$. In this case, $\prox_g = {\mathcal T}_{\lambda}$ with ${\mathcal T}_{\lambda}$ being the soft thresholding operator with parameter $\lambda$ (see footnote on page \pageref{footnote_soft}).

\section{Experiments and Numerical Results}
\label{sec:experiments}
In this section, we describe experiments and numerical results comparing the different algorithms introduced in Section \ref{sec:algorithms} on the task of phase retrieval.
\subsection{Experimental Setup}
Given measurements $\bc$ as in (\ref{eq:dataAcq}), our problem consists of finding the corresponding $\bx_0$. We focus on the setting in which $\bx_0$ is known to be sparse. We vary  the signal size $n$ (with $J = \{1, 2, \ldots, n/2\}$), the sparsity level $K$ and the signal to noise ratio (SNR). In the following discussion, we will refer to a recovery method $\mathcal{M}$ which can be seen as a black box which takes as input a vector of measurements $\bc \in \RR_+^n$, support information $J$, sparsity level $K$, an initial estimate $\bx$ and outputs an estimate $\hat{\bx} \in \RR^n$ with $\supp(\hat{\bx}) \subseteq J$ and $\|\bx\|_0 \leq K$. One recovery experiment consists of the following:
\begin{itemize}
	\item Fix a recovery method $\mathcal{M}$, a signal length $n$, a support information set $J = \{1, 2, \ldots, n/2\}$, a sparsity level $K$ and an SNR.
	\item Generate $\bx_0 \in \RR^n$ by the following procedure:
		\begin{itemize}
			\item Choose $K$ coordinates among $J$ uniformly at random.
			\item Set these coordinate values at random in $[-4, -3] \cup [3, 4]$.
			\item Set all other coordinates to be $0$.
		\end{itemize}
	\item Generate the measurements $\bc^2 = |\FF(\bx_0)|^2 + \beps$, where $\beps$ is white Gaussian noise according to the chosen SNR. Set negative entries of $\bc^2$ to be $0$ in order to take square root.
	\item Call method $\mathcal{M}$ 100 times with data $(\bc, J, K)$ and randomly generate initial estimates to get 100 candidate solutions $\{\hat{\bx}_{it}\}_{it = 1, 2, \ldots, 100}$.
	\item Compute the best estimate $\hat{\bx}_{best}$ with $best = \argmin_{it = 1,2\ldots,100}\{ \||\FF(\hat{\bx}_{it})| - \bc\|_2^2\}$.
	\item Compare $\hat{\bx}_{best}$ and $\bx_0$ (modulo Fourier invariances) with the following metric (sign is understood coordinatewise with $\sign(0) = 0$):
		\begin{align*}
			recovery(\hat{\bx}_{best}, \bx_0) &=
				\begin{cases}
					1,& \sign(\hat{\bx}_{best}) = \sign(\bx_0)\\
					0,& {\rm otherwise}.
				\end{cases}% \\
				%NMSE(\hat{\bx}_{best}, \bx_0) &= \frac{\|\hat{\bx}_{best} - \bx_0\|_2}{\|\bx_0\|_2}
		\end{align*}
\end{itemize}
This procedure was repeated 100 times. That is, for each method, signal length, sparsity level and SNR, we have 100 signal recovery experiments, each one associated with a support recovery status. We aggregate these results by considering the recovery probability (average of $recovery(\hat{\bx}_{best}, \bx_0)$) and the median CPU usage for a single simulation (100 calls to the method with different initialization estimates). We use the same initialization for all methods by careful initialization of random seeds. All the experiments were performed on a desk station with two 3.2 GHz Quad Core Intel Xeon processors and 64GB of RAM.
\subsection{Implementation Details}
In our numerical implementation, we used the following stopping criterion.
\begin{itemize}
	\item For alternating minimization and Wirtinger methods: the difference in successive objective value less than $10^{-8}$.
	\item For GESPAR: no swap improvement.
	\item For FISTAPH: the norm of the gradient mapping less than $10^{-8}$.
\end{itemize}
The tuning of these criteria allows to balance accuracy and computational time to some extent.

The $\ell_1$ penalized problem includes a prior sparsity inducing term of the form $g(\cdot) = \lambda \|\cdot\|_1$. It is necessary to tune the $\lambda$ parameter in order to obtain meaningful results. We considered the following strategies for different methods.
\begin{itemize}
	\item For alternating minimization, we set $\lambda = 0.2$ in all experiments.
	\item For Wirtinger based method, $\lambda$ is tuned a posteriori as a function of $n$ and $k$. The experiment was conducted for $\lambda = 1, 2.15, 4.64, 10, 21.5, 46.4, 100, 215, 464$, and we report only the best experiment for each setting.
\end{itemize}
An interesting feature of alternating minimization based methods is that in our experiments, recovery performance was very consistent for different values of $\lambda$ in different settings. As a result, we chose a single value of $\lambda$ for all experiments. The tuning of $\lambda$ for Wirtinger based algorithms is practically much more difficult. In particular, we found that the best $\lambda$ was a highly dependant function of the sparsity level $K$.

Finally, we note that $\ell_0$ based priors have the sparsity level of the estimate, $K$, as a parameter. On the other hand, $\ell_1$ based priors will not necessarily produce $K$-sparse estimates. We therefore use truncation and keep the $K$ largest entries in absolute value of the last iteration.

\subsection{Numerical Results}
The performance in terms of support recovery are presented in Figure \ref{fig:recovery} with the corresponding algorithm run time in Figure \ref{fig:time}. Each point in these plots is an average over 100 simulations of the recovery process, each simulation consisting of 100 random initializations of the method considered. AM corresponds to Fienup methods with different priors, FISTA is the accelerated variant, and WIRT stands for Wirtinger.

We make the following observations from the numerical results:
\begin{itemize}
	\item For alternating minimization, there is a consistent increase in recovery performance by switching from $\ell_0$ to $\ell_1$ based regularization priors.
	\item The $\ell_1$ prior degrades the performance of Wirtinger based methods compared to the $\ell_0$ prior.
	\item FISTAPH consistently provides the best performance and is significantly faster than its competitors.
	\item Fienup with $\ell_0$ prior leads to lower performance compared to GESPAR, which was already reported in \cite{shechtman2014gespar}.
\end{itemize}
As described in the experimental section, we added noise on the squared measurements rather than on the measurements themselves. This noise model is closer to the optimization model considered for GESPAR and Wirtinger flow than model (\ref{eq:mainProblem}) which is related to problem (\ref{eq:mainProblem2Bis}). We tried to change the noise model on a subset of experiments (additive noise on the measurements rather than squared measurements), however,  the performance of the different methods was very similar. Therefore, we only report results related to squared-measurement noise model.
\begin{figure}[ht]
	\centerline{\includegraphics[width=.65\textwidth]{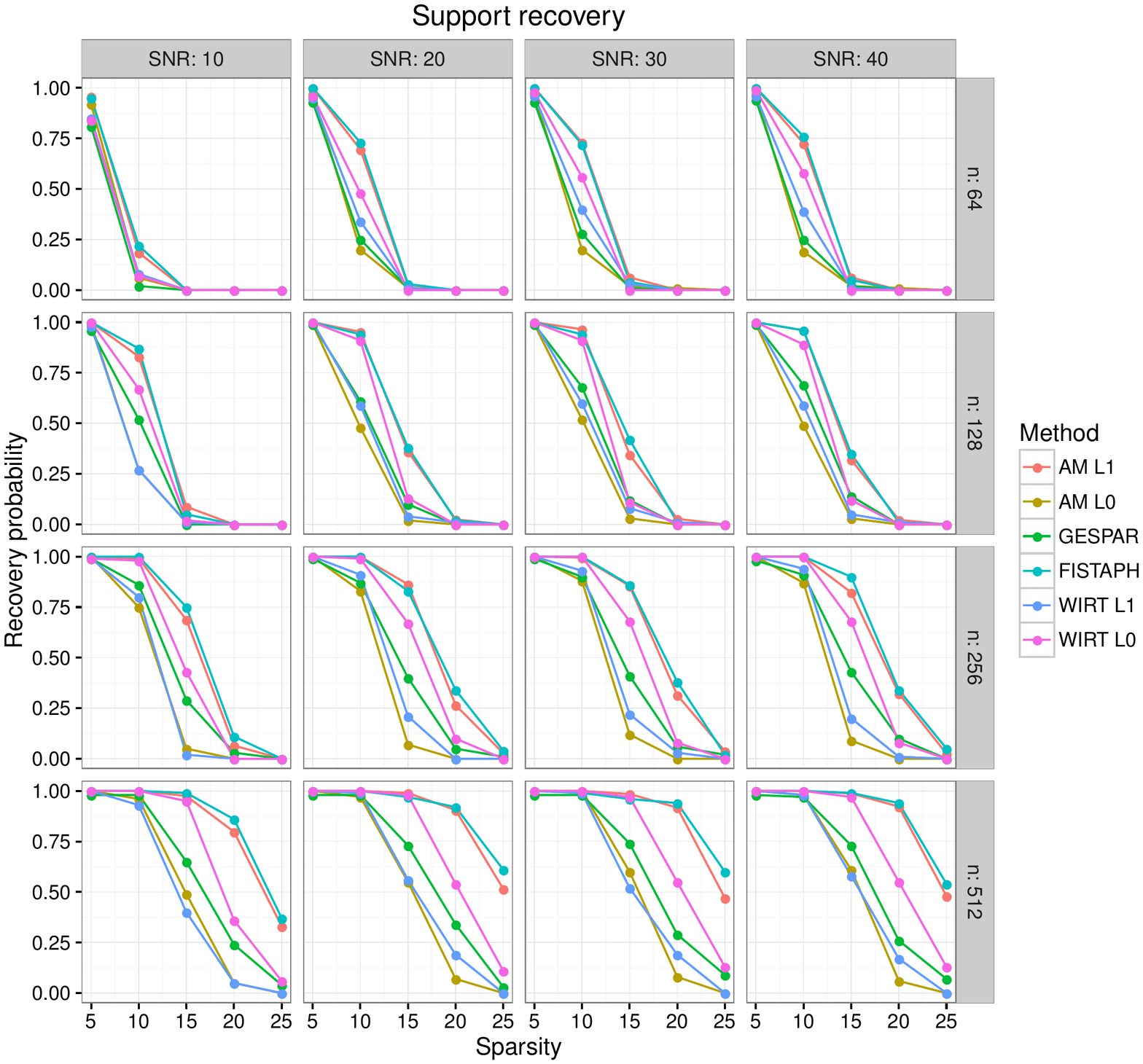}}
	\caption{Support recovery comparison. For each point, the probability is estimated based on 100 simulations. AM stands for alternating minimization and WIRT for WIRTINGER. FISTAPH is described in Section \ref{sec:convergence}, and GESPAR is the method presented in \cite{shechtman2014gespar}.}
	\label{fig:recovery}
\end{figure}

\begin{figure}[ht]
	\centerline{\includegraphics[width=.65\textwidth]{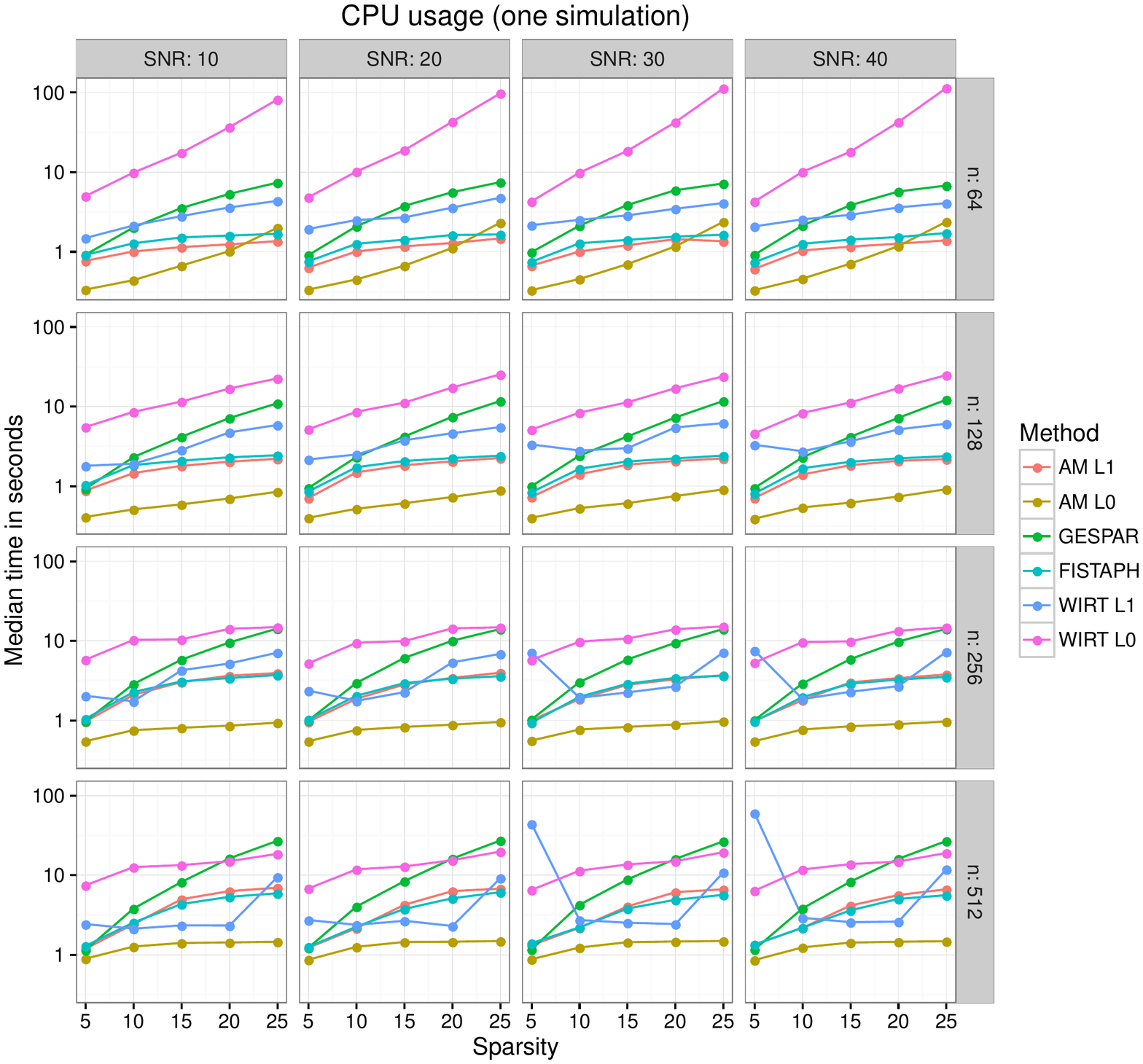}}
	\caption{Timing comparison. The ordinate axis is displayed in logarithmic scale. Each point is the median over 100 simulations, each simulation consisting in 100 random initialization for each method. AM stands for alternating minimization and WIRT for WIRTINGER. FISTAPH is described in Section \ref{sec:convergence}, and GESPAR is the method presented in \cite{shechtman2014gespar}.}
	\label{fig:time}
\end{figure}

\section{Conclusion}
The main theoretical contribution of this work is to provide a strong theoretical basis to the fact that Fienup-type methods, when used with Fourier transforms and convex priors, lead to smoothly converging sequences of estimates. This result holds under minimal assumptions and in particular, it holds globally, independently of the initialization point. Furthermore, we characterize the properties of the limiting point as Fr\'echet critical points of different optimization problems. These results shed light on important properties of one of the most well known algorithms used in the context of phase retrieval. Furthermore, based on an interpretation as a projected gradient method, we proposed a new variant of Fienup with the incorporation of a momentum term which we call FISTAPH.

On the practical side, we demonstrated based on numerical simulations that FISTAPH with $\ell_1$ regularization constitutes a very competitive alternative to other methods in the context of sparse phase retrieval.
\appendix
\section{Proof of Theorem \ref{th:convergence}}
\label{sec:proof}
The proof involves many notions of nonsmooth analysis which can be found in \cite{rockafellar1998variational}. Throughout the proof, we only consider subgradients of subdifferentially regular functions. Each subgradient can be interpreted as a Fr\'echet subgradient and the subgradient set valued mapping is closed.
	We adopt the notation of Section \ref{sec:projGrad}, letting $\bz = \bw_1 + i \bw_2$ for two real vectors $\bw_1$ and $\bw_2$ and consider the constraint set $\tilde{Z}_{\bc} = \left\{(\bw_1, \bw_2) \in \RR^n \times \RR^n;\; \bw_1 + i \bw_2 \in Z_{\bc}  \right\}$. We let $K(\bx, \bw_1, \bw_2) = \frac{1}{2} \|\bx - \bw_1\|_2^2 + \frac{1}{2} \|\bw_2\|_2^2 + g(\bx)$ be the objective function of problem (\ref{eq:mainProblem}) which with this notation becomes
	\begin{align}
		\label{eq:mainProblemBis}
		\min_{\bx \in \RR^n, (\bw_1, \bw_2) \in \tilde{Z}_{\bc}} K(\bx, \bw_1, \bw_2).
	\end{align}
	We will denote by $\delta_{\tilde{Z}_{\bc}}$, the indicator function of the set $\tilde{Z}_{\bc}$ ($0$ on the set and $+\infty$ outside). We set $\tilde{K}(\bx, \bw_1, \bw_2) = K(\bx, \bw_1, \bw_2) + \delta_{\tilde{Z}_{\bc}}(\bw_1, \bw_2)$ so that problem (\ref{eq:mainProblemBis}) is equivalent to the (unconstrained) minimization of $\tilde{K}$.
	\\

	\noindent\textbf{Proof of $(i)$:}
	Using \cite[Proposition 10.5 and Exercise 10.10]{rockafellar1998variational}, the subgradient of this nonsmooth function is of the form
	\begin{align}
		\label{eq:subGradExplicit}
		\partial \tilde{K}(\bx, \bw_1, \bw_2)
		&= \left(
		\begin{array}{c}
			\partial_{\bx}\tilde{K}(\bx, \bw_1, \bw_2)\\
			\partial_{(\bw_1, \bw_2)}\tilde{K}(\bx, \bw_1, \bw_2)
		\end{array}
		\right)\\
		&= \left(
		\begin{array}{c}
			\bx - \bw_1 + \partial g(\bx)\\
			\left(
			\begin{array}{c}
				\bw_1 - \bx\\
				\bw_2
			\end{array}
			\right) + \partial \delta_{\tilde{Z}_\bc} (\bw_1, \bw_2)
		\end{array}
		\right). \nonumber
	\end{align}
	Partial minimization over iterations yields the following
	\begin{align}
		0 &\in \bx^{k+1} - \bw^k_1 + \partial g(\bx^{k+1})\label{eq:partial1}\\
		0 &\in
			\left(
			\begin{array}{c}
				\bw^{k}_1 - \bx^{k}\\
				\bw^{k}_2
			\end{array}
			\right) + \partial \delta_{\tilde{Z}_\bc} (\bw^{k}_1, \bw^{k}_2).\label{eq:partial2}
	\end{align}
	Combining these, we have
	\begin{align}
		\label{eq:subGrad0}
		&\;\left(
		\begin{array}{c}
			0\\
			\left(
			\begin{array}{c}
				\bw^k_1 - \bx^{k+1}\\
				\bw^k_2
			\end{array}
			\right) + \partial \delta_{\tilde{Z}_\bc} (\bw^k_1, \bw^k_2)
		\end{array}
		\right)\\
		\subset&\; \partial \tilde{K}(\bx^{k+1}, \bw^k_1, \bw^k_2).\nonumber
	\end{align}
	Using (\ref{eq:partial2}),
	\begin{align}
		\label{eq:subGrad}
		\left(
		\begin{array}{c}
			0\\
			\bx^k - \bx^{k+1}\\
				0
		\end{array}
		\right)
		\in \partial \tilde{K}(\bx^{k+1}, \bw^k_1, \bw^k_2).
	\end{align}
	Finally, from strong convexity of $\tilde{K}$ with respect to its first argument, we have
	\begin{align}
		\label{eq:condition2}
		&\;\tilde{K}(\bx^{k+1}, \bw^k_1, \bw^k_2) + \frac{1}{2} \|\bx^{k+1} - \bx^k\|^2_2 \nonumber\\
		\leq&\;\tilde{K}(\bx^{k}, \bw^k_1, \bw^k_2) \nonumber\\
		\leq &\;\tilde{K}(\bx^{k}, \bw^{k-1}_1, \bw^{k-1}_2).
	\end{align}
	Since $g$ is semi-algebraic, $\tilde{K}$ is also semi-algebraic. Any semi-algebraic function satisfies the nonsmooth Kurdyka-\L ojasievicz property \cite{bolte2007clarke}. We can now use the now well established recipe \cite[Section 2.3]{attouch2013convergence} \cite[Section 3.2]{bolte2014proximal} with the two conditions (\ref{eq:subGrad}) and (\ref{eq:condition2}) to obtain that the sequence $\left\{ \|\bx^{k+1} - \bx^k\|_2 \right\}_{k \in \NN}$ is summable. This proves statement $(i)$ (convergence holds by Cauchy criterion).
	\\

	\noindent\textbf{Proof of $(ii)$:}
	Using the fact that $\tilde{K}$ has compact sublevel sets, the sequence $\left\{ (\bx^{k+1}, \bw_1^k, \bw_2^k) \right\}_{k\in \NN}$ is bounded and hence has a converging subsequence. We fix an accumulation point $(\bx^*, \bw_1^*, \bw_2^*)$ of the sequence (note that $\bx^*$ is given by $(i)$). We remark that, thanks to (\ref{eq:subGrad0}) and the fact that $\|\bx^{k+1} - \bx^k\| \to 0$, any accumulation point of the sequence is a critical point of $\tilde{K}$. Furthermore, since $\bx^k \to \bx^*$, we have using (\ref{eq:subGrad0}) that
	\begin{align*}
		-\left(
		\begin{array}{c}
			\bw^*_1 - \text{prox}_g(\bw_1^*)\\
			\bw^*_2
		\end{array}
		\right) \in \partial \delta_{\tilde{Z}_\bc} (\bw^*_1, \bw^*_2).
	\end{align*}
	This is actually the criticality condition for problem (\ref{eq:mainProblem3:a}) which proves statement $(ii)$.

\bibliography{refs}
\end{document}